\documentclass[11pt]{article}
\usepackage{amsmath,amsfonts,amsthm,amssymb, graphicx,longtable, color, algorithm, algorithmic}
\usepackage[shortlabels]{enumitem}
\usepackage[hmargin={0.8in,0.80in},vmargin={1.0in,1.0in}]{geometry}
\usepackage{setspace}
\usepackage{float}
\usepackage{subfigure}
\usepackage{array}
\usepackage{soul}
\usepackage{multirow}
\usepackage{bbm}
\usepackage[pdftex,colorlinks=true,urlcolor=blue,citecolor=black,anchorcolor=black,linkcolor=black]{hyperref}

\providecommand{\U}[1]{\protect\rule{.1in}{.1in}}
\providecommand{\U}[1]{\protect\rule{.1in}{.1in}}
\theoremstyle{plain}
\newtheorem{theorem}{Theorem}
\newtheorem{proposition}{Proposition}

\newtheorem{assumption}{Assumption}
\newtheorem{remark}{Remark}
\newtheorem{definition}{Definition}

\renewenvironment{proof}[1][Proof]{\noindent\textbf{#1.} }{\ \rule{0.5em}{0.5em}}

\begin{document}
\title{PERFECT SAMPLING OF MULTIVARIATE HAWKES PROCESSES}

\author{Xinyun Chen and 
	Xiuwen Wang\thanks{Xiuwen Wang is also with the Shenzhen Research Institute of Big Data.}\\ [12pt]
	School of Data Science\\
	The Chinese University of Hong Kong, Shenzhen\\
	2001 Longxiang road, Longgang District\\ Shenzhen, Guangdong, 518172, CHINA\\
}
\date{July 12, 2020}
\maketitle
\begin{abstract}
As an extension of self-exciting Hawkes process, the multivariate Hawkes process models counting processes of different types of random events with mutual excitement. In this paper, we present a perfect sampling algorithm that can generate i.i.d. stationary sample paths of multivariate Hawkes process without any transient bias. In addition, we provide an explicit expression of algorithm complexity in model and algorithm parameters and provide numerical schemes to find the optimal parameter set that minimizes the complexity of the perfect sampling algorithm.
\end{abstract}
\section{Introduction}\label{sec: intro}
\label{sec:intro}
As an extension of self-exciting Hawkes process, the multivariate Hawkes process has been popularly used in literature to model the counting process of different types of random events because of its ability to capture mutual excitation. In finance, multivariate Hawkes process is used to examine market returns and investor sentiment interactions \cite{yang2018applications} and predict cyber-attacks frequency \cite{bessy2020multivariate}. In medical science, it is used to fit clinical events in patient treatment trajectory \cite{duan2019clinical}. In population dynamics, it is adopted in modeling the lapse risk in life insurance \cite{barsotti2016lapse}. In social media, retweet cascades \cite{rizoiu2017tutorial} and popularity evolution of mobile Apps \cite{ouyang2018modeling} can be well predicted using this process.

There are some previous works on simulation of multivariate Hawkes processes. The early works mostly rely on the intensity-based thinning method such as \cite{ogata1998space} and \cite{liniger2009multivariate}. Improved methods are proposed in \cite{giesecke2007estimating} and \cite{dassios2013exact} that can sample from inter-arrival times directly without generating intensity paths. However, previous works are not focused on the stationary distribution of the Hawkes process and can not be directly applied to generate unbiased samples from the steady state of the Hawkes process. To the best of our knowledge, our paper is the first perfect sampling algorithm for multivariate Hawkes process, i.e. sample from the exact stationary distribution without approximation. In this paper, we extend the perfect sampling algorithm of univariate Hawkes process in \cite{Chen_2020} to multivariate Hawkes processes with mutual excitements. The key component of our algorithm is to represent the Hawkes process as a Poisson-cluster process where each cluster follows a multivariate branching process. In addition, we derive an explicit expression of the algorithm complexity as a function of the model and algorithm parameters. In particular, we show that the complexity function is convex in algorithm parameters and hence there is a unique optimal choice of algorithm parameters that minimizes the algorithm complexity, which can be numerically computed using standard convex optimization solvers. 

The rest of the paper is organized as follows. In Section \ref{sec: model}, we introduce the cluster representation of multivariate Hawkes processes and our assumptions. In Section \ref{sec: algorithm Hawkes}, we introduce the perfect sampling algorithm and derive the complexity function. In Section \ref{sec: numerical}, we implement the perfect sampling algorithm and test the performance of the simulation algorithm. 

\section{Model and Assumptions}\label{sec: model}
\subsection{Hawkes Process}\label{subsec: Hawkes}
A $d$-variate Hawkes process is a $d$-dimension counting process $N(t)=(N_1(t), N_2(t),\ldots, N_d(t))$ that satisfies 
\begin{equation*}
P(N_i(t+\Delta t)-N_i(t)=m|\mathcal{F}(t))=
\begin{cases}
\lambda_i(t)\Delta t + o(\Delta t), & m=1\\
o(\Delta t),& m>1\\
1-\lambda_i(t)\Delta t + o(\Delta t)& m=0,
\end{cases}
\end{equation*}
as $\Delta t\to 0$, where $\mathcal{F}(t)$ is the associated filtration and $\lambda(t)=(\lambda_1(t),\ldots,\lambda_d(t))\in\mathbb{R}_+^d$ is called the \textit{conditional intensity} such that
\begin{equation}\label{eq: Hawkes intensity}
\lambda_i(t) = \lambda_{i0} + \sum_{j=1}^d\sum_{\{k:t^j_k<t\}}h_{ji}(t-t_{k}^j),\ i=1,\ldots,d,
\end{equation} 
where for each $i\in\{1, 2, ..., d\}$, $\{t^i_k:-\infty<k<+\infty\}$ are the arrival times in direction $i$ and the constant $\lambda_{i0} >0$ is called the \textit{background intensity} in direction $i$, and for all $i, j \in \{1, 2, ..., d\}$, the function $h_{ji}:\mathbb{R}_+\to\mathbb{R}_+$ is integrable and called the \textit{excitation function}.

From the conditional intensity function \eqref{eq: Hawkes intensity}, we see that an arrival in any direction $j$ will increase the future intensity functions $\lambda_i$  for all $i$ through the non-negative excitation function $h_{ji}(t)$. Example of such excitation function like $\alpha e^{-\beta t} $ with $\alpha,\beta>0$ is widely used in practice. Therefore, arrivals of a multivariate Hawkes process are mutually-exciting. According to \cite{Hawkes1974}, for the linear Hawkes process, when the excitation functions are non-negative, the counting process $N(t)$ can be represented as a sum of independent processes, or clusters. Now we introduce another equivalent definition (or construction) of Hawkes process which represents it as clusters of branching processes.

\begin{definition}\label{def: Hawkes cluster}(\textbf{Cluster Representation of Multivariate Hawkes})  Consider a (possibly infinite) $T\geq 0$ and define a sequence of events $\{t^i_{n}\leq T\}$ in direction $i$, $i=1,\ldots,d$, according to the following procedure:
	\begin{enumerate}
		\item For direction $i=1,\ldots,d$, a set of immigrant events $\{\tau^i_{m}\leq T\}$  arrive according to a Poisson process with rate $\lambda_{i0}$ on $[0,T]$.
		\item For each immigrant event $\tau^i_{m}$, define a cluster $C_m^i$, which is the set of events brought by immigrant $\tau^i_m$. We index the events  in $ C_{m}^i$ by $k\geq 1$ and represent it by a tuple $e^i_{m,k}=(k,d^i_{m,k},p^i_{m.k},t^i_{m,k})$, where $d^i_{m,k}\in \{1, 2, ..., d\}$ is  direction at which the event occurs, $p^i_{m,k}$ is the index of its parent event, and $t^i_{m,k}$
		is its arrival time. Following this representation, we denote the immigrant event as $(1,i,0,\tau^i_{m})$.  
		\item Each cluster $C^i_m$ is generated following a branching process. In particular, upon the arrival of each event (say, in direction $d$), its direct offspring in direction $j$ will start to arrive according to a inhomogeneous Poisson process with rate function $h_{dj}(\cdot)$. (We will provide the complete procedure to generate a cluster in Algorithm \ref{alg: cluster} at the end of this section.)
		\item Collect the sequence of events $\{t^i_{n}\} = \overset{d}{\underset{ j=1}{\cup}}\overset{\infty}{\underset{m=1}{\cup}} \{t_{m,k}^j: k=1,..., |C^j_m|,d^j_{m,k}=i\}$.
	\end{enumerate}
	Then, in each direction $i$, counting process $N^i(t)$ corresponding to the event sequence $\{t_n^i\}$ is equivalent to the Hawkes process defined by conditional intensity function \eqref{eq: Hawkes intensity}.
\end{definition}

Now we introduce more random variables related to the Hawkes process that will be useful in our simulation algorithm and briefly analyze their distributional properties.
\begin{itemize}
	\item  We define for each cluster $C_m^i$ its \textit{cluster length} $L_m^i$ as the amount of time elapsed between the immigrant event and the last event in this cluster, i.e.
	\begin{equation}\label{eq: L}
	L_m^i \triangleq \max_{1\leq k\leq |C_m^i|}t^i_{m,k} - \tau_m^i.
	\end{equation}  
	We call $\tau_m^i$ the \textit{arrival time} of the cluster $C_m^i$.
	\item Let $\bar{h}_{ij} = \int_0^\infty h_{ij}(t)dt$. According to Definition \ref{def: Hawkes cluster}, multivariate Hawkes process has self-similarity property. That is, upon its arrival, an immigrant in direction $i$ will bring new arrivals in direction $j$ according to a time-nonhomogeneous Poisson process with rate $h_{ij}(t)$ for $j=1,2,\ldots,d$ and each arrival will bring other new arrivals following the same pattern. And $\bar{h}_{ij}$ is the expected number of next generation in direction $j$ brought by each event in direction $i$. 
	\item For any event $t^i_{m,k}$ in $C^i_{m}$ that is not an immigrant, and the arrival time of its parent is $t^i_{m,p^i_{m,k}}$ following our notation, we define the \textit{birth time} $b^{i}_{m,k}$ of event $t^i_{m,k}$ as 
	\begin{equation}\label{eq: b}
	b^{i}_{m,k} \triangleq t^i_{m,k}-t^i_{m,p^i_{m,k}}.
	\end{equation}
	Following the property of inhomogeneous Poisson process, the birth time of the events in direction $j$ whose parent in direction $l$ are independent identical distributed random variables that follow the probability density function $f_{lj}(t)\triangleq h_{lj}(t)/\bar{h}_{lj}$ for $t\geq 0$ conditional on number of events.
	\item Given the clear meaning of $\bar{h}_{ij}$ and $f_{ij}(\cdot)$ in the cluster representation of Hawkes process, in the rest of the paper, we shall denote by $(\mathbf{\lambda_{0}}, \bar{h}, f(\cdot))$ as the parameters that decide the distribution of a d-variate Hawkes process, where $\bar{h}=(\bar{h}_{ij})_{1\le i,j\le d}$ is a $d\times d$ matrix, $f(\cdot)=(f_{ij}(\cdot))_{1\le i,j\le d}$ is a $d\times d$ matrix-value function,  and $\mathbf{\lambda_{0}}=(\lambda_{0,i})_{1\le i\le d}$ is a vector in $\mathbb{R}^d$. Now, we could provide the complete procedure for Step 3 of Definition \ref{def: Hawkes cluster} in Algorithm \ref{alg: cluster}.
	\begin{algorithm}[H]
		\caption{Generate a cluster $C^i_m$ with parameter $(\bar{h},f(\cdot))$}
		\label{alg: cluster}
		\begin{algorithmic}
			\REQUIRE parameters $(\bar{h}, f(\cdot))$, an immigrant event $e^i_{m,1}=(1,i,0,\tau_m^i)$, time horizon $T$.
			\ENSURE  $C_m^i=\{e^i_{m,k},k=1,\ldots,|C^i_m|\}$ \\
			\STATE 1. Initialize $k=1$, $C_m^i=\{(1,i,0,\tau^i_m)\}$ and $n=|C^i_m|$.\\
			\STATE 2. While $n\ge k$, generate the next-generation events of $e^i_{m,k}=(k,d^i_{m,k},p^i_{m,k},t^i_{m,k})\in C^i_{m}$: \\
			~~~~~ For direction $j=1,\ldots,d$:\\
			~~~~~~~ i). Generate the number of children $\Lambda_j\sim Poisson(\bar{h}_{d^i_{m,k}j})$;\\
			~~~~~~~ ii). If $\Lambda_j\neq0$, generate birth times $b^i_{m,n+1},\ldots,b^i_{m,n+\Lambda_j}\overset{i.i.d}{\sim}f_{d^i_{m,k},j}(t)$, $t\in[0,T-t^i_{m,k}]$.\\
			~~~~~~~ iii). $C^i_{m}=C^i_m\cup\{e^i_{m,n+s}=(n+s,j,k,t^i_{m,k}+b^i_{m,n+s}):s=1,\ldots,\Lambda_j\}.$\\
			~~~~~ Update $n=|C^i_m|$ and $k=k+1.$	\\  		
			\STATE 3. Return $C^i_m$.
		\end{algorithmic}
	\end{algorithm}
\end{itemize}
\begin{remark}
	The iteration in Algorithm \ref{alg: cluster} will stop in finite time almost surely as long as the Hawkes process is stationary (i.e. Assumption \ref{asp: stable} holds), and thus the corresponding branching process will become extinct in finite time.
\end{remark}
\subsection{Stationary Hawkes Process} For the multivariate Hawkes process with parameters $(\mathbf{\lambda_{0}}, \bar{h}, f(\cdot))$ to be stable in long term, intuitively, each cluster should contain a finite number of events on average. According to \cite{Hawkes1971}, we shall assume the following condition holds throughout the paper,
\begin{equation}\label{eq: stable}
\max\{|\lambda|:\lambda \text{ is the eigenvalue of } \bar{h}\}<1.
\end{equation} 
Indeed, we can directly construct a stationary Hawkes process using the cluster representation as follows. First, in each direction $i$, we extend the  homogeneous Poisson process $\{\tau_m^i\}$ of the immigrants, or equivalently, the cluster arrivals, to time interval $(-\infty, \infty)$. For this two-ended Poisson process, we index the sequence of immigrant arrival times by $\{\pm 1, \pm 2,...\}$ such that $\tau^i_{-1}\leq 0<\tau^i_{1}$ and generate the clusters $\{C^i_{\pm m }: i=1,\ldots,d;m=1,2,\ldots\}$ independently for each $m$ following the procedure in Definition \ref{def: Hawkes cluster}. Then, the events that arrive after time $0$ form a stationary sample path of a Hawkes process on $[0,\infty)$. In detail, for each $1\leq j\leq d$, let
\begin{equation*}
\begin{aligned}
N_j(t) &\triangleq |\cup _{i=1}^d\cup _{m=-\infty}^{\infty }\{(k,d^i_{m,k},p^i_{m.k},t^i_{m,k}):k=1,2,...,|C^i_m|,d^i_{m,k}=j, 0\leq t^i_{m,k}\leq t\}|\\ 
\end{aligned}
\end{equation*}
then $N(t)=(N_1(t),...,N_d(t))$ is a stationary multivariate Hawkes process.\\

Our goal is to simulate the stationary multivariate Hawkes process efficiently and we close this section by summarizing our technical assumptions on the model parameters $(\lambda_0, \bar{h}, f(\cdot)))$:
\begin{assumption}\label{asp: stable}
	The stability condition \eqref{eq: stable} holds.
\end{assumption}
\begin{assumption}\label{asp: light tail}
	There exists $\eta>0$ such that $\int_0^{+\infty}e^{\eta t}h_{ij}(t)dt<\infty$ for all $1\le i\le j\le d$.
\end{assumption}
\begin{assumption}\label{asp: tilting}
	For at least one $\eta>0$ such that Assumption~\ref{asp: light tail} holds, we can sample from the probability density function $p(t)\propto h_{ij}(t)e^{\eta t}$ for all $1\le i\le j\le d$.
\end{assumption}

	\begin{remark}
		Assumption 1 is natural. Assumption 2 basically requires that mutual-excitement effects are short-memory, i.e. the excitation function is exponentially bounded. Actually, this assumption is satisfied in most of the application papers cited in the Introduction. Assumption 3 is assumed to guarantee algorithm implementation. Since the main focus of our algorithm is on the design of the importance sampling and the acceptance-rejection scheme, we simply assume an oracle is available to simulate from the tilted distribution. In real application, users can resort to a variety of simulation methods to generate the tilted distribution.
\end{remark}

\section{Perfect Sampling of Hawkes Process}\label{sec: algorithm Hawkes}
Following \cite{Rasmussen} and \cite{Chen_2020}, we can decompose the counting process of the stationary Hawkes in each direction $N_j(t)$ as
\begin{equation}\label{eq: stable N}
\begin{aligned}
N_j(t) &\triangleq |\cup _{i=1}^d\cup _{m=-\infty}^{\infty }\{e^i_{m,k}=(k,d^i_{m,k},p^i_{m.k},t^i_{m,k}):k=1,2,...,|C^i_m|,d^i_{m,k}=j, 0\leq t^i_{m,k}\leq t\}|\\ 
&=\sum_{i=1}^d  |\cup _{m=-\infty}^{-1}\{e^i_{m,k}\in C^i_m:d^i_{m,k}=j, 0\leq t^i_{m,k}\leq t\}| +   \sum_{i=1}^d|\cup _{m=1}^{\infty}\{e^i_{m,k}\in C^i_m:d^i_{m,k}=j,0\leq t^i_{m,k}\leq t\}| \\
&\triangleq N_{0,j}(t) + N_{1,j}(t),
\end{aligned}
\end{equation}
where $N_{0,j}(\cdot)$ are corresponding to those events that come from clusters that arrived before time $0$ and depart after time $0$, while $N_{1,j}(\cdot)$ are brought by immigrants arrive on $[0,T]$. Since the immigrants arrive on $[0,T]$ according to homogeneous Poisson processes, $\{N_{1,j}(\cdot):1\leq j\leq d \}$ can be simulated directly following the procedure described in Definition \ref{def: Hawkes cluster}. So the key step in the perfect sampling algorithm is to simulate the clusters that arrive before time $0$ and depart after time $0$. With slight abusing of notation, we shall denote the set of such clusters by $N_0$ and let $N_0 =\cup_{i=1}^d N_0^i$ where $N_0^i$ are the set of clusters brought by immigrants in direction $i$ before time 0 and lasting after time 0. 

For each $1\leq i\leq d$, let $p_i(t)=P(L^i_m>t)$ be the probability that a cluster brought by an immigrant in direction $i$ last for more than $t$ units of time. Then, by Poisson thinning theorem, the clusters in $N^i_0$ arrive on time horizon $(-\infty,0]$ according a in-homogeneous Poisson process with intensity function $\gamma_i(t) = \lambda_{0,i}p_i(-t)$, for $t\leq 0$. Besides, a cluster in $N^i_0$ should follow the conditional distribution of the branching process $C^i_m$ described in Definition  \ref{def: Hawkes cluster} on the event that $\{L_m^i>-\tau_m^i\}$.

Following Definition \ref{def: Hawkes cluster}, $N_0^1,\ldots,N_0^d$ are independent and can be simulated separately. In our new algorithm, we extend Algorithm 1 proposed in \cite{Chen_2020}  to sample, for each $i$, from the conditional distribution of $C_m^i$ given $L_m^i>t$ without evaluating the function $p_i(t)$.  Also, we shall improve the efficiency of simulating the conditional distribution of $C_m^i$ by importance sampling and acceptance-rejection methods. In particular, we define the \textit{total birth time} as
$$B^i_m\triangleq \sum_{k=2}^{|C_m^i|} b_{m,k}^i,$$
where $b_{m,k}^i$ is the birth time of event $k$ in cluster $C_m^i$ as defined in \eqref{eq: b}, and $b_{m,1}^i=0$ for the immigrant event. We design a change of measure based on exponential tilting of $C_m^i$ with respect to $B_m^i$ which is called \textit{importance distribution}. To do this, let us characterize such importance distribution by analyzing the cumulant generating function (c.g.f.) of $B_m^i$.

\begin{proposition}\label{prop: B}
	For all $i\in\{1,\ldots,d\}$, consider a cluster  $C_m^i$ in direction $i$ with parameter $(\bar{h},f(\cdot))$. Let $L_m^i$ and $B_m^i$ be its cluster length and total birth time, respectively. Then, the following statements are true:
	\begin{enumerate}[(1)]
		\item $B_m^i\geq L_m^i$ for all $1\leq i\leq d$.
		\item Under Assumption \ref{asp: light tail}, there exists $\theta_0>0$ such that for any $0<\theta<\theta_0$, the cumulant generating function (c.g.f.) of $B_m^i$ is well-defined for all $1\leq i\leq d$, i.e. $\psi_{B^i}(\theta)\triangleq\log E[\exp(\theta B_m^i)]<\infty$. Besides, for $0\leq \theta\leq \theta_0$, the vector $(\psi_{B^i}(\theta))_{i=1}^d$  satisfies a system of equations:
		\begin{equation}\label{eq: B cgf}
		\psi_{B^i}(\theta)=\sum_{j=1}^{d}\bar{h}_{ij} \left( \exp(\psi_{f_{ij}}(\theta)+\psi_{B^j}(\theta))-1\right) , \forall 1\leq i\leq d,
		\end{equation}
		with $\psi_{f_{ij}}(\theta) = \log(\int_0^\infty e^{\theta t}f_{ij}(t)dt)$ being the c.g.f. of the birth time $b^i_j$ as defined in \eqref{eq: b} for all $1\leq i, j\leq d$.
		\item Let $\mathbb{P}$ be the probability distribution of a cluster $C_m^i$. Let $\mathbb{Q}$ be the importance distribution of the cluster under exponential tilting by parameter $0<\eta<\theta_0$ with respect to the total birth time $B_m^i$, i.e.
		$$d\mathbb{Q}(C_m^i) =\exp(\eta B_m^i-\psi_{B^i}(\eta))\cdot d\mathbb{P}(C_m^i).$$
		Then, sampling a cluster from the importance distribution $\mathbb{Q}$ is equivalent to sampling a cluster in direction $i$ with parameter $(\bar{h}_{lj}\exp(\psi_{f_{lj}}(\eta)+\psi_{B^j}(\eta)), f_{lj,\eta}(\cdot))_{1\le l,j\le d}$ with $f_{lj,\eta}(t) = f_{lj}(t) \cdot\exp(\eta t - \psi_{f_{lj}}(\eta))$.
	\end{enumerate} 
\end{proposition}
\begin{proof}[Proof of Proposition \ref{prop: B}]
	
	We prove the three statements of Proposition \ref{prop: B} one by one.
	
	(1). Given $i$, to see $B_m^i\geq L_m^i$, it is intuitive to represent the cluster as a tree ignoring the direction of each event. Let the immigrant event $\tau_m^i$ be the root node and link each event $t^i_{m,k}$ to its parent event by an edge of length $b^i_{m,k}$. Then $B_m^i$ is equal to the total length of all edges in the tree while $L_m^i$ is equal to the length of the longest path(s) from the root node to a leaf node. Therefore, $B_m^i\geq L_m^i$.
	
	(2). By definition, for each $i$, $B_m^i=\sum_{k=2}^{|C_m^i|}b_{m,k}^i$, where $b_{m,k}^i$ follows probability density $f_{lj}(t)$ for some $1\leq k,j\leq d$. 
	Let 
	$$\tilde{b}_{k} = \sum_{1\leq l, j\leq d} \tilde{b}_{ij, k}, \text{ with }\tilde{b}_{lj, k}\stackrel{i.i.d.}{\sim } f_{lj}(t), \forall 1\leq l, j\leq d.$$ 
	Then we have 
	$$E[\exp(\theta B_m^i)]\leq E\left[\exp\left(\theta\sum_{k=2}^{|C_m^i|}\tilde{b}_k\right)\right].$$
	Note that $\sum_{k=2}^{|C_m^i|}\tilde{b}_k$ is  the sum of $|C_m^i|-1$ i.i.d. random variables. On the one hand, it is known in literature \cite{MultiHawkesMGF} that the c.g.f. of $E[\exp(\theta |C_m^i|)]$ is finite in a neighbourhood around $0$. On the other hand, following Assumption \ref{asp: light tail}, $$E[\exp(\theta \tilde{b}_k)]= \prod_{1\leq l,j\leq d}\int_0^\infty e^{\theta t}f_{ij}(t)dt<\infty, \forall \theta\leq \eta.$$
	As a consequence, we can conclude that $E\left[\exp\left(\theta\sum_{k=2}^{|C_m^i|}\tilde{b}_k\right)\right]$ is finite in a neighbourhood around the origin and so is $E[\exp(\theta B_m^i)]$. As this is true for all $i$, we can find $\theta_0>0$ such that the c.g.f.'s of $B_m^i$ for all $i$ are finite for all $0<\theta<\theta_0$.
	
	Now, for fixed $\theta\in (0,\theta_0)$, let's characterize $(\psi_{B^i}(\theta))_{i=1}^d$. For a cluster in direction $i$, we denote by $K^i_j$ the number of direct offsprings in direction $j$ brought by the immigrant in direction $i$. Then, $K^i_j$ is a Poisson random variable with mean $\bar{h}_{ij}$. Based on the self-similarity property of the branching process, a direct child event in direction $j$ of the immigrant will bring a sub-cluster following the same distribution as a cluster $C^j$, and thus the total birth time of this sub-cluster follows the same distribution as $B^j$. For all $1\leq i, j\leq d$, let $b^i_j(k)$ be i.i.d. copies of birth time following the probability density $f_{ij}(t)$ and $B^j(k)$ be i.i.d. copies of total birth time $B^j$, for $k=1,2,...,K^i_j$. Then, we have
	\begin{equation*}\label{eq: prop B}
	\begin{aligned}
	\psi_{B^i}(\theta) &= \log(E[\exp(\theta B^i_m)]) = \log\left(E\left[\exp\left(\theta\sum_{k=1}^{|C_m^i|} b^i_{m,k}\right)\right]\right)\\
	&= \log\left(E\left[E\left[\exp\left(\theta\sum_{j=1}^d\sum_{k=1}^{K^i_j} \left(b^i_j(k) +B^{j}(k)\right)\right)| K^i_1,K^i_2,\ldots,K^i_d \right]\right]\right)\\
	&= \log\left(E\left[\exp\left( \sum_{j=1}^dK^i_j(\psi_{f_{ij}}(\theta) + \psi_{B^j}(\theta)) \right)\right]\right)\\
	&=\sum_{j=1}^d \bar{h}_{ij} [\exp(\psi_{f_{ij}}(\theta) + \psi_{B^j}(\theta))-1].
	\end{aligned}
	\end{equation*}
	
	(3). Under $\mathbb{Q}$, the c.g.f. of $B^i_m$ becomes $\psi_{B^i,\eta}(\theta)=\psi_{B^i}(\theta+\eta)-\psi_{B^i}(\eta)$. We can compute,
	\begin{equation*}
	\begin{aligned}
	\psi_{B^i,\eta}(\theta)
	&=\sum_{j=1}^d \bar{h}_{ij} [\exp(\psi_{f_{ij}}(\theta+\eta) + \psi_{B^j}(\theta+\eta))-1]-\sum_{j=1}^d \bar{h}_{ij} [\exp(\psi_{f_{ij}}(\eta) + \psi_{B^j}(\eta))-1]\\
	&=\sum_{j=1}^d \bar{h}_{ij} [\exp(\psi_{f_{ij}}(\theta+\eta) + \psi_{B^j}(\theta+\eta))-\exp(\psi_{f_{ij}}(\eta) + \psi_{B^j}(\eta))]\\
	&=\sum_{j=1}^d \bar{h}_{ij} \exp(\psi_{f_{ij}}(\eta) + \psi_{B^j}(\eta))[\exp(\psi_{f_{ij},\eta}(\theta) + \psi_{B^j,\eta}(\theta))-1]\\
	\end{aligned}
	\end{equation*}
	with $\psi_{f_{ij},\eta}(\theta)=\psi_{f_{ij}}(\eta+\theta)-\psi_{f_{ij}}(\eta)$ be the c.g.f.  corresponding to probability density function $f_{ij,\eta}(\cdot)$. The above calculation shows that $\psi_{B,\eta}(\theta)$ equals exactly to the c.g.f. of the total birth time  of a cluster with parameter $(\bar{h}_{ij}\exp(\psi_{f_{ij}}(\eta)+\psi_{B^j}(\eta)), f_{ij,\eta}(\cdot))_{j=1,\ldots,d}$. Therefore, to simulate the exponentially tilted distribution of $C^i$ with parameter $\eta$ is equivalent  to generate a cluster such that number of direct child events in direction $j$ of an event in direction $i$ is a Poisson with mean $\bar{h}_{ij}\exp(\psi_{f_{ij}}(\eta)+\psi_{B^j}(\eta))$, and their birth time following the density function $f_{ij,\eta}(t)$.
\end{proof}\\

Given Proposition \ref{prop: B}, we know how to simulate the importance distribution of $C_m^i$. Now we provide our perfect sampling algorithm as follows.\\

\begin{algorithm}[H]
	\caption{Simulating $N_0$}
	\label{alg: hawkes}
	\begin{algorithmic}
		\REQUIRE parameters of the Hawkes Process $(\lambda_{0}, \bar{h}, f(\cdot))$, a positive vector $\eta=(\eta_1, \eta_2, ..., \eta_d)\in[0,\theta_0]^d$
		\ENSURE  $N_0 =\{C_m^i: m\leq -1, L_m^i>-\tau_m^i, i=1,\ldots,d\}$ \\
		\STATE 1. Initialize $N_0=\{\}$.
		\STATE 2. For direction $i=1,\ldots,d$:\\
		\begin{enumerate}[1).]
			\item Compute $\psi_{f_{lj}}(\eta_i)$ and $\psi_{B^j}(\eta_i)$ from Equation \eqref{eq: B cgf}, for $1\leq l,j\leq d$
			\item generate an in-homogeneous Poisson process on $(-\infty, 0]$ with rate function $$\tilde{\gamma}_i(t) = \lambda_{0,i}\exp(\psi_{B^i}(\eta_i)+\eta_i t),~ t\leq 0,$$
			~and obtain arrivals $\{\tau^i_1,...,\tau^i_{K}\}$.
			\item For $m = 1,2,...,K$:
			\begin{enumerate}[a).]
				\item Generate a cluster $C_m^i$ with parameter $(\bar{h}_{lj}\exp(\psi_{f_{lj}}(\eta_i)+\psi_{B^j}(\eta_i)), f_{lj,\eta_i}(\cdot))_{1\leq l, j\leq d}$ by Algorithm\ref{alg: cluster} and $U_m\sim U[0,1]$.
				\item Accept $C_m^i$ and update $N_0 = N_0 \cup \{C_m^i\}$ if both of the following conditions are satisfied: 
				\begin{enumerate}[i.]
					\item $L_m^i>-\tau_m^i$,
					\item $U_m\leq \exp(-\eta_i(B_m^i+\tau_m^i))$.
				\end{enumerate}
			\end{enumerate}
			
		\end{enumerate}
		
		\STATE 3. Return $N_0$.
	\end{algorithmic}
\end{algorithm}

Algorithm \ref{alg: hawkes} actually contains two importance-sampling steps. When simulating the immigrants, i.e. the arrivals of clusters, it simulates an in-homogeneous Poisson with a larger intensity function $\tilde{\gamma}_i(t)\geq \gamma_i(t)$. Now we show that the output of Algorithm \ref{alg: hawkes} follows exactly the distribution of $N_0$, and provide an explicit expression of algorithm complexity as a function of model and algorithm parameters.

\begin{theorem}
	\label{thm1}
	The list of clusters generated by Algorithm \ref{alg: hawkes} exactly follows the distribution of $N_0$. In particular, for each direction $i\in\{1,\ldots,d\}$,
	\begin{enumerate}[(1)]
		\item  The arrival times of the clusters follow an in-homogeneous Poisson process  with intensity $\gamma_i(t)=\lambda_{0,i}p_i(-t)$ for $t\in(-\infty,0]$. 
		\item For each cluster $C_m^i$ in the list, given its arrival time $\tau_m^i$, it follows the conditional distribution of a cluster given that the cluster length $L^i_m>-\tau_m^i$.   
	\end{enumerate}
	Besides, the expected total number of random variables generated by Algorithm \ref{alg: hawkes} before termination is a \textit{convex} function in $\eta\in[0,\theta_0]^d$ with explicit expression as follows.
	\begin{equation}\label{eq: complexity}
	X(\eta)=\sum_{i=1}^d \frac{\lambda_{0i}\exp(\psi_{B^i}(\eta_i))}{\eta_i}(1+\tilde{S}_i(\eta_i)), 
	\end{equation} 
	where $\tilde{S}_i(\eta) $ is the $i$-th row sum of the matrix $\tilde{S}(\eta) = (I-\tilde{h}(\eta))^{-1}$ and $\tilde{h}(\eta)=[\tilde{h}_{ij}(\eta)]_{1\leq i,j\leq d}$  with $\tilde{h}_{ij}(\eta)=\bar{h}_{ij}\exp(\psi_{f_{ij}}(\eta)+\psi_{B^j}(\eta))$ for all $1\leq i, j\leq d$. 
\end{theorem}

\begin{proof}[Proof of Theorem \ref{thm1}] To prove Statement (1) for each direction $i$,
	by Poisson thinning theorem, it suffices to show that, for each $m$, the acceptance probability of the cluster $C_m^i$ equals to $$\frac{\gamma_i(\tau_m^i)}{\tilde{\gamma}_i(\tau_m^i)} = \mathbb{P}(L_m^i>-\tau_m^i)\exp(-\eta\tau_m^i)/\exp(\psi_{B^i}(\eta)).$$ According to Proposition \ref{prop: B}, the importance distribution $\mathbb{Q}$ and the target distribution $\mathbb{P}$ satisfies 
	$$d\mathbb{Q}(B_m^i = x) = d\mathbb{P}(B_m^i=x)\cdot\frac{\exp(\eta x)}{\exp(\psi_{B^i}(\eta))} .$$
	Therefore, the probability for cluster $C_m^i$ to be accepted in Step 3  is
	\begin{equation*}
	\begin{aligned}
	&E_\mathbb{Q}\left[1\left(L^i_m>-\tau_m\text{ and }U_m<\exp(-\eta (B_m^i+ \tau_m^i))\right)\right]\\
	=~& \int  1(L_m^i>-\tau_m^i) \exp(-\eta (B^i_m+ \tau_m^i))d\mathbb{Q}\\ 
	=~& \int  1(L_m^i>-\tau_m^i) \exp(-\eta (B^i_m+ \tau_m^i)) d\mathbb{P}\cdot\frac{\exp(\eta B^i_m)}{\exp(\psi_{B^i}(\eta))} \\
	= ~&\int  \frac{\exp(-\eta\tau_m^i) }{\exp(\psi_{B^i}(\eta))}1(L_m^i>-\tau_m^i) d\mathbb{P}\\
	=~&\mathbb{P}(L_m^i>-\tau_m^i)\exp(-\eta\tau_m^i)/\exp(\psi_{B^i}(\eta)).
	\end{aligned}
	\end{equation*}
	Therefore, we obtain Statement (1).
	
	From the above calculation, we can also see that, given $m$ and $\tau_m^i$, and any event $A\in \sigma(C_m^i)$, the joint probability
	\begin{align*}
	P(C_m^i\in A, C_m^i\text{ is accepted})& = \int  \frac{\exp(-\eta\tau_m^i) }{\exp(\psi_{B^i}(\eta))}1(C_m^i\in A, L_m^i>-\tau_m^i) d\mathbb{P}\\
	& \propto \mathbb{P}(L_m^i>-\tau_m^i, C_m^i\in A).
	\end{align*}
	Therefore, the accepted sample of $C_m^i$ indeed follows the conditional distribution of $C_m^i$ given $\{L_m^i> -\tau_m^i\}$, and we obtain Statement (2).
	
	To check \eqref{eq: complexity}, we first note that the expected total number of random variables generated by Algorithm \ref{alg: hawkes} can be expressed as $\sum_{i=1}^d M_i S_i$, where $M_i$ is the expected total number of clusters in direction $i$,  and $S_i$ is the average number of events in each cluster. Following Algorithm \ref{alg: hawkes}, the number of those clusters in direction $i$ is a Poisson random variable with mean 
	$$\int_0^\infty \tilde{\gamma}_i(-t)dt = \frac{\lambda_{0,i}\exp(\psi_{B^i}(\eta_i))}{\eta_i}.$$
	Now we compute the average number of events one cluster brought by an immigrant in direction $i$. Note that in Algorithm \ref{alg: hawkes}, we use importance distribution of parameters $(\bar{h}_{lj}\exp(\psi_{f_{lj}}(\eta_i)+\psi_{B^j}(\eta_i)), f_{lj,\eta_i}(\cdot))_{1\leq l, j\leq d}$ to generate the clusters in direction $i$. Let $S_{lj}$ be the expected number of events in direction $j$  brought by an immigrant in direction $k$ under this set of parameters. Then, by the self-similarity of the branching process, we have
	\begin{equation}\label{eq:Sij}
	S_{lj}=\mathbbm{1}(l=j)+\sum_{k=1}^d\tilde{h}_{lk}(\eta_i)S_{kj}, ~\forall 1\leq l, j\leq d,
	\end{equation}
	with $\tilde{h}_{lk}(\eta_i)=\bar{h}_{lk}\exp(\psi_{f_{lk}}(\eta_i)+\psi_{B^k}(\eta_i))$. 
	As $\eta_i<\theta_0$,  the moment generating function of the clusters under the importance distribution is finite in a neighbourhood around 0 and the corresponding branching processes are all stable. Then the determinant $|\tilde{h}(\eta_i)-I|\neq 0$, so the inverse matrix $\tilde{S}(\eta_i)=(I-\tilde{h}(\eta_i))^{-1}$ is well defined. Therefore, 
	$$S_i =\tilde{S}_i(\eta_i)\triangleq\sum_{j=1}^d \tilde{S}_{ij}(\eta_i) \text{ is the row sum of the matrix} (I-\tilde{h}(\eta_i))^{-1}.$$
	As a consequence, the expected total number of  random variables is 
	$$\sum_{i=1}^d M_i S_i=\sum_{i=1}^d \frac{\lambda_{0i}\exp(\psi_{B^i}(\eta_i))}{\eta_i}(1+\tilde{S}_i(\eta_i)).$$
	
	Finally, what remains to be proved is that the function $X(\eta)=\sum_{i=1}^d \frac{\lambda_{0i}\exp(\psi_{B^i}(\eta_i))}{\eta_i}(1+\tilde{S}_i(\eta_i))$ is convex. This suffices to show that for each $1\leq i\leq d$,  $x(\eta_i)\triangleq\frac{\lambda_{0i}\exp(\psi_{B^i}(\eta_i))}{\eta_i}(1+\tilde{S}_i(\eta_i))$ is convex. In the rest of proof, for convenience, we write $k_0=i$,  $g(\eta_i)\triangleq\lambda_{0i}\exp(\psi_{B^i}(\eta_i))/\eta_i$ and denote the product of functions $f(x)g(x)$ as $(f*g)(x)$. As the eigenvalue of the matrix $\tilde{h}(\eta_i)<1$, we can write,
	\begin{align*}
	x(\eta_i)&=g(\eta_i)\left(1+\sum_{j=1}^d \tilde{S}_{ij}(\eta_i)\right)
	=2g(\eta_i)+\sum_{m=1}^{\infty}\sum_{k_1=1}^d\sum_{k_2=1}^d\ldots\sum_{k_m=1}^d(g*\tilde{h}_{ik_1}*\ldots*\tilde{h}_{k_{m-1}k_m})(\eta_i).
	\end{align*}
	For any $m\ge 1$, let $a_m(\eta_i)\triangleq\sum_{l=0}^m\psi_{B^{k_l}}(\eta_i)+\sum_{l=1}^m\psi_{f_{k_{l-1}k_l}}(\eta_i)$. Then, we have
	\begin{equation*}
	\begin{aligned}
	\dfrac{\partial}{\partial\eta_i}(g*\tilde{h}_{ik_1}*\ldots*\tilde{h}_{k_{m-1}k_m})(\eta_i)
	=&\tilde{h}_{ik_1}\tilde{h}_{k_1k_2}\ldots\tilde{h}_{k_{m-1}k_m}\dfrac{\exp(\psi_{B^i}(\eta_i))}{\eta_i^2} \left(   \eta_ia_m'(\eta_i) -1\right),
	\end{aligned}
	\end{equation*}
	and 
	\begin{align*}
	\dfrac{\partial^2}{\partial\eta_i^2}(g*\tilde{h}_{ik_1}*\ldots*\tilde{h}_{k_{m-1}k_m})(\eta_i)
	=\tilde{h}_{ik_1}\tilde{h}_{k_1k_2}\ldots\tilde{h}_{k_{m-1}k_m}\dfrac{\exp(\psi_{B^i}(\eta_i))}{\eta_i^3}\left((\eta_ia_m'(\eta_i)-1)^2+\eta_i^2a_m''(\eta_i)+1\right).
	\end{align*}
	By H$\overset{..}{\text{o}}$lder's inequality, the cumulant generating function $\psi$ is convex. Therefore, we have $$a_m''(\eta_i)=\sum_{l=0}^m\psi_{B^{k_l}}''(\eta_i)+\sum_{l=1}^m\psi_{f_{k_{l-1}k_l}}''(\eta_i)>0.$$ Now we can conclude that $(g*\tilde{h}_{ik_1}*\ldots*\tilde{h}_{k_{m-1}k_m})(\eta_i)$ is convex for any $m\ge1$ and $k_1, k_2,...,k_m\in\{1, 2, ...,d\}$. Similarly, we can check that $g''(\eta_i)>0$. Besides, by Assumption \ref{asp: stable} the row sums of matrix $\tilde{h}$, $\sum_{j=1}^d\tilde{h}_{ij}\le\underset{i=1,\ldots,d}{\max}\sum_{j=1}^d\tilde{h}_{ij}\triangleq u<1$. By Assumption \ref{asp: light tail} and the smoothness of $\psi$, for any given $\eta_i$, there exists a neighbourhood $E_i$ round $\eta_i$, such that for all $\eta\in E_i$, the quantities $$U_1=\max\limits_{j=1,\ldots,d}|\psi'_{B^j}(\eta)|,~ U_2=\underset{1\le l,k\le d}{\max}|\psi'_{f_{lk}}(\eta)|,~ U_3=\underset{j=1,\ldots,d}{\max}|\psi_{B^j}''(\eta)|,\text{ and }U_4=\underset{1\le l,k\le d}{\max}|\psi''_{f_{lk}}(\eta)|,$$
	are all finite. Then we can check that, for any $\eta\in E_i$, the infinite sums
	\begin{align*}
	&\sum_{m=1}^{\infty}\sum_{k_1=1}^d\sum_{k_2=1}^d\ldots\sum_{k_m=1}^d\frac{\partial}{\partial \eta}(g*\tilde{h}_{ik_1}*\ldots*\tilde{h}_{k_{m-1}k_m})(\eta)
	\leq \frac{\exp(\psi_{B^i}(\eta))}{\eta_i^2}\sum_{m=1}^{\infty} u^m\left( a m-1\right)<\infty\\
	&\sum_{m=1}^{\infty}\sum_{k_1=1}^d\sum_{k_2=1}^d\ldots\sum_{k_m=1}^d\frac{\partial^2}{\partial \eta_i^2}(g*\tilde{h}_{ik_1}*\ldots*\tilde{h}_{k_{m-1}k_m})(\eta_i)
	\leq \dfrac{\exp(\psi_{B^i}(\eta_i))}{\eta_i^3}\sum_{m=1}^{\infty} u^m\left( a^2 m^2+bm+2\right)<\infty,
	\end{align*} 
	where $a=\eta(U_1+U_2)$, $b=\eta^2(U_3+U_4)-2\eta(U_1+U_2)$. This justifies the exchange of derivative with infinite sum. Therefore, we can conclude $x''(\eta_i)>0$ for $i=1,2,...,d$  and $X(\mathbf{\eta})$ is convex.
\end{proof}	

\begin{remark}
	The complexity result \eqref{eq: complexity} not only guarantees that Algorithm \ref{alg: hawkes} terminates in finite time in expectation. As it is a convex function in algorithm parameter $\mathbf{\eta}$, one can apply some standard convex optimization procedure to find out the optimal $\mathbf{\eta}$ to minimize the computational cost of the simulation algorithm.
\end{remark}

\vskip 2ex

\section{Numerical Experiments}\label{sec: numerical}
We implement Algorithm \ref{alg: hawkes} in Python to test the performance and correctness of our perfect sampling algorithm. For a $d$-variate Hawkes Process with parameters $(\mathbf{\lambda_{0}},\bar{h},f(\cdot))$, its stationary intensity rate is $$[\sum_{k=1}^{d}\lambda_{0,k}S_{k,i}]_{i=1,\ldots,d}=(I-\bar{h}^T)^{-1}\mathbf{\lambda_{0}}.$$
In the following numerical experiments, we shall compare the theoretical value of intensity rate with the simulation estimation of $E(N(1))$ from 10000 rounds of perfect sampling algorithms. To test the complexity result \eqref{eq: complexity}, we also compare the average number of random variables generated in each simulation round with the theoretical value computed from \eqref{eq: complexity}. 

\paragraph{A symmetric 2-dimension case} We first consider a 2 dimensional Hawkes process with excitation function $h_{ij}(t)=\alpha_{ij}e^{-\beta_{ij}t}$, $\lambda_{0,i}=1$, $1\leq i,j\leq 2$ and\begin{align*}
[\alpha_{ij}]=\left( \begin{array}{cc}
1&2\\
2&1
\end{array}\right),~[\beta_{ij}]=\left( \begin{array}{cc}
2&8\\
8&2
\end{array}\right).  
\end{align*} Then, the stationary intensity rate of this Hawkes process is $[4,4]^T$ and the optimal $\eta^*\approx0.0664$ computed from \eqref{eq: complexity}. (As the process is symmetric, the optimal $\eta_i^*$ are equal in both directions $i=1,2$. )
We  then use Algorithm \ref{alg: hawkes} to simulate a stationary Hawkes process on time interval $[0,1]$ for different choices of $\eta$ (in both directions). If our algorithm is correct, the average number of events generated on this time interval should almost equal to the stationary intensity, i.e. $E[N_1(1)]=E[N_2(1)]=4$, regardless of the choice of algorithm parameter $\eta$. 

We report our simulation results for different values of $\eta$ in Table \ref{tb: Hawkes2}. In all cases, the 95\% confidence interval constructed from simulation data covers the true value of $4$. Besides, we see that the average random numbers (\#rvs) generated in one simulation round are close to the theoretical values and reach the minimum when $\eta$ is mostly close to the $\eta^*$.
\begin{table}[H]
	\centering
	\begin{tabular}{|c|c|c|c|c|}
		\hline
		$\eta$ & 95\%  Confidence Interval & \#rvs&Theoretical \#rvs&CPU time(s)\\
		\hline
		0.03&$4.0091\pm0.0699,~4.0343\pm0.0704$&395.8266&395.3016&0.030869\\
		0.05&$4.0153\pm0.0698,~4.0498\pm0.0718$&278.7656&279.6228&0.022252\\
		0.06 &$3.9975\pm0.0704,~3.9812\pm0.0687$ &261.9655 &260.4849 &0.021223 \\
		0.07&$4.0115\pm0.0693,~4.0417\pm0.0700$&258.1993&258.5722&0.021221\\
		0.08&$4.0301\pm0.0712,~ 4.0258\pm0.0711$&279.726&280.3890&0.022979\\
		0.09&$4.0272\pm0.0701,~ 4.0239\pm 0.0700$&372.3206&372.1390&0.030785\\
		\hline
	\end{tabular}
	\caption{95\% confidence interval for $E[N_i(1)],i=1,2$ and the average number of random variables generated in one simulation round, estimated from 10000 i.i.d. sample path of Hawkes process on $[0,1]$ simulated by Algorithm \ref{alg: hawkes} with different $\eta$.}
	\label{tb: Hawkes2}
\end{table}
\paragraph{Non-symmetric 5-dimension case}
Here we apply the Algorithm \ref{alg: hawkes} in a non-symmetric Hawkes process with higher dimension, and compare its performance with the naive method by simulating the Hawkes process for a long time to approximate the steady state. The parameters are  $\mathbf{\lambda_0}=[0.1,0.2,0.1,0.3,0.4]^T$ and\begin{align*}
[\alpha_{ij}]=\left( \begin{array}{ccccc}
0.8 & 0.8 & 0.2 & 0.8 & 1.    \\
0.8 & 0.1 & 0.9 & 0.1 & 0.5   \\
0.5 & 0.6 & 0.7 & 0.5 & 0.3   \\
0.2 & 0.9 & 0.9 & 0.7 & 0.4   \\
0.3 & 0.2 & 0.2 & 0.9 & 1.1 \\
\end{array}\right),~[\beta_{ij}]=\left( \begin{array}{ccccc}
4.9 & 4.1 & 4.9 & 3.3 & 3.3 \\
3.3 & 4.1 & 4.9 & 1.7 & 3.3 \\
7.3 & 5.7 & 4.9 & 7.3 & 5.7 \\
0.9 & 5.7 & 2.5 & 8.1 & 7.3 \\
6.5 & 3.3 & 3.3 & 7.3 & 4.9\\
\end{array}\right).  
\end{align*} Then, the stationary intensity rate of this Hawkes process is $[0.5640, 0.5534, 0.6163, 0.6860, 0.9346]^T$ and $\eta^*\approx[0.1234, 0.1306, 0.1405, 0.1234, 0.1378]$.

\begin{figure}[H]
	\centering
	\includegraphics[width=10cm]{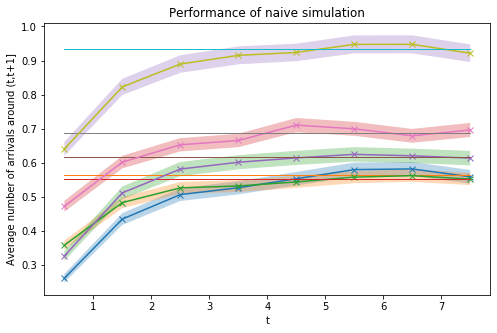}
	\caption{Estimated intensity rates on $(t,t+1]$ and 95\% confidence intervals from naive simulation of Hawkes process in 10000 simulation rounds (while the horizontal lines are the theoretical intensity rates).}
\end{figure}

Since we know the true value of the steady-state mean in this special case, Figure 1 implies that naive simulation requires at least 7 units of time for the process to reach the steady state. Therefore, we next report the result of naive simulation using the number events on $(6,7]$. In Table \ref{tb: Hawkes}, we compare the estimation result and algorithm complexity, measured by the total number of random variables generated, of the perfect sampling and naive algorithms. The results show that the complexity of our perfect sampling algorithm is larger than, but of the same magnitude as the naive simulation.  However, it is important to keep in
mind that there is no theoretical result that how many iterations in naive simulation is needed for the transient multivariate Hawkes process converge to its stationary version. So, in practice, extra computational cost is needed to diagnose whether the processes have been well stabilized or not, and thus the actual complexity might be much higher. In contrast, the samples of the perfect sampling algorithm are entirely unbiased with theoretical guarantee.

\begin{table}[H]	
	\centering
	\begin{tabular}{|c|c|c|c|c|}
		\hline
		Method& 95\%  Confidence Interval & \#rvs &Theoretical \#rvs&CPU time(s)\\
		\hline
		Alg.1&$\begin{array}{ccccc}
		0.5702&  0.5491&0.6056 &0.6808  &0.9264\\
		\pm&\pm&\pm&\pm&\pm\\
		0.0203,&0.0176,&0.0208,&0.0204,&0.0261
		\end{array}$ &56.9082&56.8234&0.0145\\
		Naive&$\begin{array}{ccccc}
		0.5816& 0.5618 &0.6201 &0.6788  &0.9474\\
		\pm&\pm&\pm&\pm&\pm\\
		0.0205	,&0.0182,&0.0210,&0.0203,&0.0264
		\end{array}$	
		&38.9925&-&0.0063\\
		\hline
	\end{tabular}
	\caption{Comparison results between perfect sampling and naive simulation by 10000 rounds of simulation. The perfect sampling algorithm is implemented with optimal $\eta^*$.}
	\label{tb: Hawkes}
\end{table}

\section{CONCLUSION}
In this paper we develop the first efficient perfect sampling algorithm that can generate stationary sample paths of multivariate Hawkes process exactly. An immediate application of our algorithm is in numerical computation for the steady state, or long-run average, of stochastic models that involves multivariate Hawkes processes, for instance, queuing networks with Hawkes arrivals. Our algorithm can be used to generate stationary arrivals to those models, and help reduce the estimation error caused by the transient bias in the Hawkes process distribution.  A possible direction of future research is to combine our algorithm with existing perfect sampling algorithms for queuing networks to develop perfect sampling algorithms for queuing networks with Hawkes arrivals.

\bibliographystyle{plain}
\bibliography{references}

\section*{AUTHOR BIOGRAPHIES}
\noindent {\bf XINYUN CHEN} is an Assistant Professor in the School of Data Science at the Chinese University of Hong Kong, Shenzhen. She received her Ph.D in Operations Research from Columbia University in 2014. Her research interests include applied probability, Monte Carlo method and their applications in financial markets. She has published papers in journals including Annals of Applied Probability, Mathematics of Operations Research and  Queuing Systems. Her email address is \href{mailto:chenxinyun@cuhk.edu.cn} {chenxinyun@cuhk.edu.cn} and her website is  \href{https://myweb.cuhk.edu.cn/chenxinyun}{https://myweb.cuhk.edu.cn/chenxinyun}.\\

\noindent {\bf XIUWEN WANG} is a Ph.D. student in the School of Data Science at the Chinese University of Hong Kong, Shenzhen and Shenzhen Research Institute of Big Data. She earned her bachelor degree in the School of Mathematic at Sun Yat-sen University. Her research interests include stochastic simulation design and analysis, simulation based optimization currently. Her email address is \href{mailto:wangxiuwen@link.cuhk.edu.cn}{wangxiuwen@link.cuhk.edu.cn}.

\end{document}